\documentclass{llncs}
\usepackage{tikz-cd}
\usepackage{adjustbox,lipsum}
\usetikzlibrary{arrows}
\usepackage{listings}
\usepackage{amsmath,epsfig,epstopdf, amssymb}
\usepackage{algpseudocode}
\algnewcommand\algorithmicforeach{\textbf{for each}}
\algdef{S}[FOR]{ForEach}[1]{\algorithmicforeach\ #1\ \algorithmicdo}

\usepackage{wrapfig}
\usepackage{lipsum}
\usepackage{float}
\floatstyle{boxed}
\newfloat{algorithm}{htbp}{loa}
\floatname{algorithm}{Algorithm}
\newfloat{algorithm}{htbp}{loa}
\usepackage{wrapfig}
\newcommand{\rimp}{\Rightarrow}
\usepackage{multirow}

\usepackage{etoolbox}\AtBeginEnvironment{algorithmic}{\small‌​}

\newtheorem{thm}{Theorem} 
\newtheorem{observation}[thm]{Observation}

\title{Efficient LTL Decentralized Monitoring Framework Using Formula Simplification Table}

\author{Omar Bataineh$^{\star}$, David Rosenblum$^{\star}$, and Mark Reynolds$^{\ast}$ }

\institute{$^{\star}$National University of Singapore \\
         $^{\ast}$University of Western Australia}

\date{}

\begin{document}

\maketitle

\begin{abstract}

This paper presents a new technique for optimizing
formal analysis of propositional logic formulas and 
Linear Temporal Logic (LTL) formulas,
namely the formula simplification table. 
A formula simplification table is a mathematical table that 
shows all possible simplifications of the formula
 under different truth assignments of its variables.
The advantages of constructing a simplification table of a formula are two-fold.
First, it can be used to compute the logical influence 
weight of each variable in the formula,
which is a metric that shows the importance of the variable 
in affecting the outcome of the formula.
Second, it can be used to identify variables that have the highest logical 
influences on the outcome of the formula.
We demonstrate the effectiveness of formula
simplification table in the context of software verification 
by developing efficient framework to the well-known 
decentralized LTL monitoring problem.

\end{abstract}

\section{Introduction}

This paper describes several new techniques to improve
formal analysis of both propositional
logic formulas and Linear temporal logic formulas.
The new presented improvement techniques are mainly based on the notion
of formula simplification table.
A formula simplification table is a mathematical table that shows
all possible simplified forms of the formula under
different truth assignments of its variables.
Constructing a simplification
table of a formula has several advantages.
First, it can be used to compute a logical influence 
weight of each variable in the formula,
which is a metric that shows the importance of the 
variable to the outcome of the formula.
Second, it can be used to identify variables in the specification
that have the highest logical influence on its outcome.
Third, it can be used to synthesize Boolean expressions 
for sets of configurations (i.e., assignments of variables) 
that yield the same simplified formulas of the original formula.
Hence,  formula simplification
table can be used to optimize existing solutions
of several fundamental software verification problems.


However, the scalability of formula simplification table requires 
controlling  the size of the formula
(i.e., the number of variables in the formula),
as the size of the table grows exponentially with respect to the number of variables.
To address this issue we present an algorithm for
reducing large formulas to a simplified form
by detecting and contracting variables whose
logical influences on the outcome of the formula are equivalent.
Instead of using specialized heuristics
to control formula size, we present a systematic approach for simplifying LTL formulas
that identifies variables with equivalent logical influences
on the outcome of the formula. 
Hence, simplifications we perform in this paper cannot be obtained by detecting
duplicates, syntactic contradictions or tautologies. 


The presented simplifications are mainly based 
on the observation that most of large formulas
contain variables with equivalent logical influences,
and therefore one needs not to consider all the variables in the formula
when constructing a formula simplification table. 
It is possible then to construct
much smaller formula sufficient to prove the original property.
In particular, given an input formula $\varphi$, 
our simplification technique produces a simplified formula $\varphi^{'}$
while reducing and contracting variables
whose logical influences on the outcome of the formula are equivalent.
Then some sound logical extension rules
are applied to draw valid conclusions about the original formula.

We demonstrate the effectiveness of formula simplification
table  in the context of software verification 
by developing efficient solution to the well-known decentralized LTL monitoring problem.
In decentralized LTL monitoring problem, a set of processes 
cooperate with each other in order to monitor a global LTL formula,
where each process observes only subset of the variables
of the main formula. The problem is to allow each process
to monitor the formula through communicating with other processes.
The goal is then to develop a solution that allows processes
to detect violation of the global formula as early as possible
and with least communication overhead.
We develop an efficient solution to the problem 
by synthesizing efficient communication strategy for processes
that allows them to propagate their observations in an optimal way.

\section{The Decentralized LTL Monitoring Problem}

A distributed program $\mathcal{P} = \{p_1,p_2,...,p_n \}$ is a set of $n$ processes which cooperate with each other in order to achieve a certain task. 
Distributed monitoring is less developed and more challenging than local monitoring: they involve designing a distributed algorithm that monitors another distributed algorithm.
In this work,  we assume that no two processes share a common variable.  Each process of the distributed system emits events at discrete time instances. Each event $\sigma$ is a set of actions denoted by some atomic propositions from the set $AP$. We denote $2^{AP}$ by $\Sigma$ and call it the alphabet of the system. We assume that the distributed system operates under the perfect synchrony hypothesis, and that each process sends and receives messages at \textit{discrete} instances of time, which are represented using identifier $t \in \mathbb{N}^{\geq 0}$. An event in a process $p_i$, where $1 \leq i \leq n$, is either

\begin{itemize}

\item  internal event (i.e. an assignment statement),

\item  message sent, where the local state of $p_i$ remains unchanged, or

\item  message received,  where the local state of $p_i$ remains unchanged. 

\end{itemize}

Since each process sees only a projection of an event to its locally observable set of actions, we use a projection function $\Pi_i$ to restrict atomic propositions to the local view of monitor $\mathcal{M}_i$ attached to process $p_i$, which can only observe those of process $p_i$. For atomic propositions (local to process $p_i$), $\Pi_i: 2^{AP} \rightarrow 2^{AP}$, and we denote $AP_i = \Pi_i (AP)$, for all $i =1...n$. For events, $\Pi_i :2^{\Sigma} \rightarrow 2^{\Sigma}$ and we denote $\Sigma_i = \Pi_i (\Sigma)$ for all $i= 1...n$. We assume that $\forall_{i, j \leq n, i \neq j}  \rimp AP_i \cap AP_j = \emptyset$ and consequently  $\forall_{i, j \leq n, i \neq j} \rimp \Sigma_i \cap \Sigma_j = \emptyset$. That is, events are local to the processes where they are monitored.
The system's global trace, $g = (g_1, g_2,..., g_n)$ can now be described as a sequence of pair-wise unions of the local events of each process's traces. We denote the set of all possible events in $p_i$ by $E_i$ and hence the set of all events of $P$ by $E_P = \bigcup_{i=1}^{n} E_i$. Finite traces over an alphabet $\Sigma$ are denoted by $\Sigma^{*}$, while infinite traces are denoted by $\Sigma^{\infty}$.

\begin{definition}(\textbf{LTL formulas \cite{Pnueli1977}}). The set of LTL formulas is inductively defined by the grammar
\[
\varphi ::=  true \mid p \mid \neg \varphi \mid \varphi \lor \varphi \mid X \varphi \mid F \varphi \mid G \varphi \mid \varphi U \varphi  
\]
where $X$ is read as next, $F$ as  eventually (in the future), 
$G$ as  always (globally), $U$ as until, and $p$ is a propositional variable.

\end{definition}

\begin{definition} (\textbf{LTL Semantics \cite{Pnueli1977}}). Let $ w = a_0 a_1.. . \in \Sigma^{\infty}$ be a infinite word with $i \in N$
being a position. Then we define the semantics of LTL formulas inductively as follows

\begin{itemize}

\item $w, i \models true$

\item $w, i \models \neg \varphi$ iff $w, i \not\models \varphi$

\item $w, i \models p$ iff $ p \in a_i$
 
\item  $w, i \models \varphi_1 \lor \varphi_2$ iff $w, i \models \varphi_1$ or $w, i \models \varphi_2$

\item $w, i \models F \varphi$ iff $ w, j \models \varphi$ for some $j \geq i$

\item $w, i \models G \varphi$ iff $ w, j \models \varphi$ for all $j \geq i$

\item $w, i \models \varphi_1 U \varphi_2$ iff $\exists_{k \geq i}$ with $w, k \models \varphi_2$ and $\forall_{i \leq l < k}$ with $ w, l \models \varphi_1$

\item  $w, i \models X \varphi$ iff $w, i+1 \models  \varphi$

\end{itemize}

\end{definition}

We now review the definition of three-valued semantics LTL$_3$ that is used to interpret common LTL formulas,
as defined in \cite{Bauer2011}. The semantics of  LTL$_3$ is defined on finite prefixes to obtain a truth value from the set $\mathbb{B}_3 = \{ \top, \bot, ?  \}$.

\begin{definition} (\textbf{LTL$_3$ semantics}). Let $u \in \Sigma^{*}$ denote a finite word.  The truth value of a LTL$_3$ formula $\varphi$ with respect to $u$, denoted by $[u \models \varphi]$, is an element of $\mathbb{B}_3$ defined as follows:

$$
[u \models \varphi]  = 
\begin{cases}
\top & \textrm{if $\forall \sigma\in \Sigma^{\infty} : u\sigma \models \varphi $} \\
 \bot & \textrm{if  $\forall \sigma\in \Sigma^{\infty} : u\sigma \not\models \varphi $} \\
 ? & otherwise
\end{cases}
$$

\end{definition}

According to the semantics of LTL$_3$ the outcome of the evaluation of $\varphi$ can be inconclusive (?). This happens if the so far observed prefix $u$ itself is insufficient to determine how $\varphi$ evaluates in any possible future continuation of $u$.

\begin{problem} (\textbf{The decentralized monitoring problem}).
Given a distributed program $\mathcal{P} = \{p_1, p_2,..., p_n\}$, a finite global-state trace $\alpha \in \Sigma^{*}$, an $LTL$ property $\varphi$, and a set of monitor processes $ \mathcal{M} = \{ M_1, M_2,..., M_n\}$ such that

\begin{itemize}

\item  monitor $M_i$ can read the local state of process $p_i$, and

\item  monitor $M_i$ can communicate with other monitor processes.


\end{itemize}
The problem is then to design an algorithm that allows each monitor $M_i$ to evaluate $\varphi$ through communicating with other monitor processes.
The problem can be studied under different settings and different assumptions.
However, in this work,
we make a number of assumptions
about the class of systems that can be monitored in our framework.

\end{problem}

\begin{itemize}

\item \textbf{A1}: the monitored system  is a synchronous timed system with a global clock;


\item \textbf{A2}: processes are reliable (i.e., no process is malicious).
\end{itemize}

It is interesting to note that the synchronous assumption imposed in our setting
is by no means unrealistic, as in many real-world systems,
communication occurs synchronously. We refer the reader to \cite{BauerF12,ColomboF14} 
in which the authors
discussed a number of interesting examples of protocols 
for safety-critical systems in which communication occurs synchronously. 

\section{Detecting Variables with Equivalent Logical Influences}

In this section, we discuss techniques that can be used to detect variables in a Boolean
formula or in an LTL formula whose logical influences on the outcome of the formula are equivalent.  Given a formula $\varphi$ with a set of propositional variables 
$prop(\varphi) = \{a_1,..., a_n\}$, we ask the following questions:

\begin{enumerate}

\item Does $\varphi$ contain  variables whose logical influences 
on the outcome of the formula are equivalent?

\item Can we develop tests to extract variables with equivalent logical influences?

\item Can we assign a value (a logical influence measure) to every variable in $\varphi$,
corresponding to its importance  in affecting the outcome of the formula?

\item Can we identify the variable that have the highest 
logical influence on the outcome of the formula $\varphi$?

\end{enumerate}


First we need to define what it means for variables to have equivalent logical influence. 
Consider the following simple propositional logic formula
$
\varphi = (a \lor (b \land c)).
$
Do variables $a$ and $b$ have equivalent logical influence?
Do variables $b$ and $c$ have equivalent logical influence? 
Which variable has the highest logical influence on the outcome of $\varphi$?
The answers to these questions depend on how the formula $\varphi$
is simplified under different truth assignments of its variables.
To answer questions (1-4) we introduce
what we call a formula simplification table which shows 
how the formula gets simplified under different truth assignments of the variables.
We first give a definition of formula simplification
table and then give some examples
by which we demonstrate how one can construct a simplification table for a formula.

\begin{definition} (\textbf{Formula simplification table}).
A simplification table is a mathematical table that shows 
all possible simplified forms of a given formula that result from different 
truth assignment of its variables. 
A simplification table has one column for each input variable, and one final column showing the simplified formula under the given combination of truth assignments. The variables take their truth values from the truth domain $\mathbb{B}_3= \{\bot, \top, ?\}$.  Each row of the table contains one possible configuration of the variables
and the formula that results from substituting truth values 
of the variables in the main formula.

\end{definition}

A simplification table for the formula $\varphi = (a \lor (b \land c))$
is given in Table \ref{table:EX1}. 
Before proceeding further, let us summarize the basic
rules that one needs to follow when construction a simplification table of a formula.

\begin{itemize}

\item Truth values of variables are taken from the truth domain $\mathbb{B}_3 =\{?, \bot, \top\}$.

\item Only variables with known truth values will be substituted in the formula.


\end{itemize}

\begin{table} 
 \centering
\begin{tabular}{|c|c|c|c|c|c|c|c|c|c|c|c|}  
\hline
a &b & c & Simplified formula & a &b & c & Simplified formula &a &b & c &  Simplified formula \\ 
\hline
? & ? & ? & $(a \lor (b \land c))$& $\top$ &?  &?  & $\top$& $\bot$& ? & ?& $(b \land c)$ \\ 
\hline
? & $\bot$ & ? & $a $&  $\top$ & ? & $\top$ & $\top$ & $\bot$& ?& $\bot$& $\bot$\\ 
\hline
? & $\bot$ & $\bot$ & $a$& $\top$ & ? & $\bot$ &$\top$ & $\bot$& ?& $\top$& $b$\\ 
\hline
? & $\bot$ & $\top$ &$a$ &$\top$  & $\bot$ & ? & $\top$& $\bot$& $\bot$&? & $\bot$\\ 
\hline
? & $?$ & $\top$ &$(a \lor b)$ &  $\top$& $\bot$ &$\bot$  &  $\top$& $\bot$& $\bot$& $\bot$& $\bot$\\ 
\hline
? & ? & $\bot$ &$a$ &  $\top$&  $\bot$& $\top$ &$\top$ & $\bot$& $\bot$& $\top$&$\bot$ \\ 
\hline
? & $\top$ & ? &$(a \lor c)$ & $\top$ & $\top$ & ? &$\top$ & $\bot$& $\top$& ?& $c$\\ 
\hline
? & $\top$ & $\bot$ &$a$ &  $\top$&  $\top$&  $\top$ &$\top$ & $\bot$& $\top$& $\top$& $\top$\\ 
\hline
? & $\top$ & $\top$ &$\top$ & $\top$ & $\top$ &  $\bot$& $\top$& $\bot$& $\top$& $\bot$& $\bot$\\ 
\hline
\end{tabular} 
\caption{A simplification table for the formula $\varphi = (a \lor (b \land c))$} \label{table:EX1}
\end{table}

The simplification table provides a rich source of information about the structure of formula and its simplifications under different truth assignments of its variables, that is not available from other data structures. In addition to providing key information about the importance of each variable in the formula, the table also allows one to detect variables
with equivalent logical influence and configurations that lead to the same simplified formula.
We first discuss the following two new notions: (1) variables with equivalent logical influences, and (2) the influence weight of a variable on the outcome of  the formula. 

\begin{definition} (\textbf{Variables with equivalent logical influences})\label{EquiVar}.
Two variables in a formula are said to be equivalent 
in their logical influences on the outcome 
of the formula if under the same truth assignment 
they yield formulas with identical syntactic structure.
Let $\varphi$ be a formula 
and $prop (\varphi)$ be the set of variables in $\varphi$.
We say that the two variables $a, b \in prop (\varphi)$
have equivalent logical influences on $\varphi$ 
(denoted as $a \equiv b$) if the following condition holds

$$
 \begin{array}[t] {l}
prog (\varphi, a = \bot) = rename (prog (\varphi, b = \bot), a, b) ~ \land \\
prog (\varphi, a = \top) =   rename (prog (\varphi, b = \top) , a, b)
\end{array}
$$

where $prog(\varphi, a = v)$ is a function that returns a new
formula of $\varphi$ after substituting the truth value of $a$ in $\varphi$ and $rename (prog (\varphi, b = \top) , a, b)$ is a function that replaces all instances of $a$ in $prog (\varphi, b = \top) $ to $b$ (i.e., changing the name of the variable $a$ to $b$). For example, $prog ((a \land b) , a = \top) = b$ and $rename (prog ((a \land b \land c), b = \top), a, b) = rename ((a \land c), a, b) = (b \land c)$.

\end{definition}

From the simplification table of $\varphi = (a \lor (b \land c))$ (Table \ref{table:EX1})
we note that the two variables $b$ and $c$ have equivalent logical influence on the outcome of $\varphi$  as $prog(\varphi, b = \bot) = rename (prog(\varphi, c = \bot), b , c)$
and  $prog(\varphi, b = \top) = rename (prog(\varphi, c = \top), b , c)$, 
while the variables $a$ and $b$ have inequivalent logical influence   
as  $prog(\varphi, a = \top) \neq rename (prog(\varphi, b = \top), a, b)$.

\begin{definition} (\textbf{Influence weights of variables})\label{InfluenceWight}.
The influence weight of a variable in a given formula
is a metric that shows the importance of the variable
in affecting the outcome of the formula. 
It can be computed from the simplification table
of the formula. Let $\varphi$ be a formula
and $prop(\varphi)$ be the set of variables of $\varphi$ 
and $a \in prop(\varphi)$. The influence weight of the variable $a$ 
(denoted as $IW_{\varphi}(a)$) can be computed by taking the ratio of
the number of formulas in the simplification table that $a$ appears in 
(let us denote by $f_a$) to the number  of truth combinations 
of the variables in which $a$ has unknown truth value ($a = ?$) 
(let us denote it by $C_{a = ?}$).
Hence, $IW_{\varphi}(a)$ can be computed as follows
$$
IW_{\varphi}(a) = \dfrac{f_a}{C_{a = ?}}
$$
\end{definition}

From Table \ref{table:EX1} we note that 
$IW_{\varphi}(a) = \dfrac{8}{9} $, $IW_{\varphi}(b) = \dfrac{4}{9} $,
and $IW_{\varphi}(c) = \dfrac{4}{9} $.
It is easy to see that the variable $a$ has higher logical influence
on the outcome of the formula than both $b$ and $c$.
This can be shown from the value of the influence weight of $a$
which is larger than the weights of both $b$ and $c$.
Note that the larger the influence weight of the variable, the more important the variable
(i.e., the variable has higher influence on the outcome of the formula). 
As we discuss later  there are several factors that can affect 
the influence weight of a variable in a given formula:
(a) the number of times the variable appears in the formula,
(b) the logical connectives used in the formula, and
(c) the length of the formula.




\begin{definition} (\textbf{Equivalent configurations}). 
Let $\varphi$ be a formula with a set of 
propositional variables $\{a_1,..., a_n\}$. 
We say that the two configurations $O = (a_1 = v_1,..., a_n = v_n)$
and $O^{'} = (a_1 = v_1^{'},..., a_n = v_n^{'})$ are equivalent
if they lead to the same simplified formula,
where $ (v_1,..., v_n, v_1^{'},..., v_n^{'}) \in \mathbb{B}_3$.
Formally, we say that the two configurations $O$ and $O^{'}$ are equivalent if

$$
 \begin{array}[t] {l}
 prog(...(prog(\varphi, a_1 = v_1), a_2 = v_2),...,a_n = v_n) = \\
 prog(...(prog(\varphi, a_1 = v_1^{'}), a_2 = v_2^{'}),...,a_n = v_n^{'})
 \end{array}
$$

\end{definition}

The simplification table of a formula can be used also to derive Boolean formulas
characterizing the conditions under which the main formula 
can be simplified into some specific formulas. 
Deriving such Boolean formulas can be very useful for certain problems 
in formal verification such as the decentralized LTL monitoring problem,
where processes can use such formulas to determine 
the minimal set of variables whose truth values
need to be propagated.
For example, for the formula  $\varphi = (a \lor (b \land c))$
one can see from the simplification table of $\varphi$
that there are multiple configurations that lead to the same
simplified formula. For instance, there are five different configurations
that simplify the formula to the atomic formula $\phi = a$. 
One can then derive a Boolean formula characterizing the cases
under which $\varphi$ can be simplified to $\phi$,
which will be in this case $\mathbb{B}_{\phi} = (\overline{b} + \overline{c}$).
Note that $\mathbb{B}_{\phi}$ is given here in its simplest form.

The technique can be used also for LTL formulas 
to compute the influence weights of variables in a given LTL formula. 
Note that for propositional logic formulas, we call the
table as simplification table since the formula
gets simplified once we substitute a truth value
of a variable in the formula (i.e., the size of the formula is reduced).
This is not always the case for temporal formulas, as
the formula may be expanded at each state of the trace 
to express sets of obligations (requirements)
that the system should fulfill for the remaining
part of the trace.
We therefore call the table as progression table
rather than simplification table when dealing with LTL formulas.
The key question is then how to deal with temporal operators
when constructing a progression table.
Let us construct a progression table for the temporal formula
$
\varphi = F (a \land b) \lor G (c \land d).
$

\begin{table} 
 \centering
  \adjustbox{max width=\textwidth}{
\begin{tabular}{|c|c|c|c|c|c|c|c|c|c|}  
\hline
$a^{(t)}$ &$b^{(t)}$ & $c^{(t)}$ & $d^{(t)}$&  Progressive formula & $a^{(t)}$ &$b^{(t)}$ & $c^{(t)}$ & $d^{(t)}$& Progressive formula \\ 
\hline
? & ? & ? &? & $( (a^{(t)} \land b^{(t)})\lor XF (a \land b)) \lor (c^{(t)} \land d^{(t)} \land XG(c\land d))$ &$\bot$ &? &? &? &   $ XF (a \land b) \lor ((c^{(t)} \land d^{(t)}) \land XG(c\land d))$   \\
\hline
? & ? & $\bot$ &? & $(a^{(t)} \land b^{(t)}) \lor XF (a \land b)$ & $\bot$& $\top$& ?& $\top$&  $XF (a \land b) \lor (d^{(t)} \land XG(c\land d))$    \\
\hline
? & ? & $\bot$ &$\bot$ & $(a^{(t)} \land b^{(t)}) \lor XF (a \land b)$ & $\bot$& $\bot$ &? &? &  $ XF (a \land b) \lor ((c^{(t)} \land d^{(t)}) \land XG(c\land d))$  \\
\hline
? & ? & $\top$ &$\bot$ & $(a^{(t)} \land b^{(t)}) \lor XF (a \land b)$ & $\bot$& $\top$ & ?& ?& $ XF (a \land b) \lor ((c^{(t)} \land d^{(t)}) \land XG(c\land d))$  \\
\hline
? & ? & $\bot$ &$\top$ & $(a^{(t)} \land b^{(t)}) \lor XF (a \land b)$ &$\bot$ & $\top$& ? & $\bot$&  $ XF (a \land b)$  \\
\hline
? & ? & ? &$\bot$ & $(a^{(t)} \land b^{(t)}) \lor XF (a \land b)$ & $\bot$& ?& ?& $\bot$& $ XF (a \land b)$   \\
\hline
? & ? & $\top$ &? & $(a^{(t)} \land b^{(t)})\lor XF (a \land b) \lor (d^{(t)} \land XG(c\land d))$ & $\bot$ &? &? & $\top$ &  $XF (a \land b) \lor (d^{(t)} \land XG(c\land d))$  \\
\hline
? & ? & ? &$\top$ & $(a^{(t)} \land b^{(t)})\lor XF (a \land b) \lor (c^{(t)} \land XG(c\land d))$ &$\bot$ & $\bot$& ? &$\bot$ & $ XF (a \land b)$ \\
\hline
? & ? & $\top$ &$\top$ & $(a^{(t)} \land b^{(t)})\lor XF (a \land b) \lor XG(c\land d)$ & $\bot$& $\bot$& ?& $\top$& $XF (a \land b) \lor (d^{(t)} \land XG(c\land d))$\\
\hline
? & $\bot$ & ? &? & $ XF (a \land b) \lor ((c^{(t)} \land d^{(t)}) \land XG(c\land d))$ &$\top$ &? &? &? & $b^{(t)} \lor XF (a \land b) \lor (c^{(t)} \land d^{(t)} \land XG(c\land d))$ \\
\hline
? & $\bot$ & ? &$\bot$ & $ XF (a \land b)$ & $\top$&? &? &$\bot$ &  $b^{(t)} \lor XF (a \land b)$\\
\hline
? & $\bot$ & $\bot$ &? & $ XF (a \land b)$ &$\top$ &? &? & $\top$&  $b^{(t)} \lor XF (a \land b) \lor (d^{(t)} \land XG(c\land d))$\\
\hline
? & $\bot$ & $\bot$ &$\bot$ & $ XF (a \land b)$ & $\top$& $\top$  & ?& $\top$& $\top$ \\
\hline
? & $\bot$ & $\bot$ &$\top$ & $ XF (a \land b)$ & $\top$&$\top$ & ?& ?& $\top$\\
\hline
? & $\bot$ & $\top$ &$\bot$ & $ XF (a \land b)$ & $\top$& $\top$& ?& $\bot$& $\top$   \\
\hline
? & $\bot$ & $\top$ &$\top$ & $ XF (a \land b) \lor XG (c \land d)$ &$\top$ & $\bot$ &? &? & $ XF (a \land b) \lor ((c^{(t)} \land d^{(t)}) \land XG(c\land d))$  \\
\hline
? & $\bot$ & $\top$ &$?$ & $ XF (a \land b) \lor (d^{(t)} \land XG (c \land d))$ &$\top$ &$\bot$ &? & $\bot$&  $XF(a \land b)$ \\
\hline
? & $\bot$ & $?$ &$\top$ & $ XF (a \land b) \lor (c^{(t)} \land XG (c \land d))$ & $\top$& $\bot$& ?& $\top$&   $ XF (a \land b) \lor ((c^{(t)}  \land XG(c\land d))$ \\
\hline
? & $\top$ & $\bot$ &$\bot$ & $ a^{(t)} \lor XF (a \land b)$ &? & $\top$ & $?$ &$\bot$ & $ a^{(t)} \lor XF (a \land b)$ \\
\hline
? & $\top$ & $\bot$ &? & $ a^{(t)} \lor XF (a \land b)$ &? & $\top$ & $?$ &? &  $(a^{(t)} \lor XF (a \land b)) \lor ((c^{(t)} \land d^{(t)}) \land XG(c\land d))$ \\
\hline
? & $\top$ & $?$ &$\top$ & $(a^{(t)} \lor XF (a \land b)) \lor (c^{(t)} \land XG(c\land d))$ & ? & $\top$ & $\top$ &$?$ & $(a^{(t)} \lor XF (a \land b)) \lor (d^{(t)} \land XG(c\land d))$  \\
\hline
? & $\top$ & $\top$ &$\top$ &  $ a^{(t)} \lor XF (a \land b) \lor XG (c \land d)$& ? & $\top$ & $\top$ &$\bot$ &  $ a^{(t)} \lor XF (a \land b)$  \\
\hline
? & $\top$ & $\bot$ &$\top$ &  $ a^{(t)} \lor XF (a \land b)$& &  & & &  \\
\hline
\end{tabular} }
\caption{A partial progression table for the formula $F (a \land b) \lor G (c \land d)$} \label{table:EX2}
\end{table}

Since we mainly use the progression table to measure the influence weights of the variables to the outcome of the formula, we choose to restrict the temporal
operators to specific time step $t\geq 0$ and use the classical
expansion rules to express the semantics of the operators
(i.e., $F(a) \equiv a \lor XF(a)$).
It is interesting to note that restricting temporal operators
to specific time step does not harm the analysis, it just simplifies it. 
From the definition of influence weights (see Definition \ref{InfluenceWight})
 it is sufficient then to consider the temporal operators at single step 
to compute the logical influence of variables to the outcome of the formula.

However, before constructing a progression table for the formula
we use Definition \ref{EquiVar} to detect variables in the formula
whose logical influences on the outcome of the formula are equivalent.
This would help to reduce the size of the table.
Using Definition \ref{EquiVar} we conclude
that $a \equiv b$ and $c \equiv d$ but $a \not \equiv c$.
We therefore have two sets of variables whose logical influences are equivalent: 
$E_1 = \{a, b\}$ and $E_1 = \{c, d \}$.
In this case we do not need to construct a full progression table
for the formula as $IW_{\varphi} (a)= IW_{\varphi} (b)$
and $IW_{\varphi} (c)= IW_{\varphi} (d)$.
We only need to compute the influence weights of the variables $a$ and $c$.

From the progression table of the formula 
$\varphi = F (a \land b) \lor G (c \land d)$ (Table \ref{table:EX2})
we can see that the variables $a$ and $b$ have higher logical influences
on the outcome of the formula than the variables $c$ and $d$,
where $IW_{\varphi} (a) = IW_{\varphi} (b) =  \dfrac{27}{27} = 1$ and $IW_{\varphi} (c) =  IW_{\varphi} (d) = \dfrac{18}{27} \approx 0.66$. 
This is mainly due to the semantics of the operators $F$ and $G$
and that the subformulas  $F (a \land b)$ and 
$G (c \land d)$ are connected using the logical connective $\lor$.
This leads to the conclusion that the set of
logical and temporal operators used in the formula affect the weights of the variables.


\begin{observation}
It is possible to have a variable in an LTL formula $\varphi$ whose influence weight is one. For example, for the formula $F (a \land b \land c)$
we notice that $a \equiv b \equiv c$ and that $IW_{\varphi} (a) = IW_{\varphi} (b) = IW_{\varphi} (c) = 1$.

\end{observation}

In addition to the length of the formula and the set of logical and temporal
operators used in the formula, the number of times the variable appears in the formula
 can affect its influence weight on the outcome of the formula.
Let us consider the following example to demonstrate this.

\begin{example}
Consider the following LTL formula $ \varphi = F (a \land b) \lor G (a \land c).
$
Using Definition \ref{EquiVar} we see that the formula $\varphi$ has no variables
with equivalent logical influences, where $a \not \equiv b$ and $a \not \equiv c$ and $b \not \equiv c$. Note that even the variables $a$ and $b$  appear within the scope of the $F$ operator and the variables  $a$ and $c$ appear within the scope of the $G$ operator. This is simply because the variable $a$ appears twice in the formula which makes it the most important variable in the formula.

\end{example}

\begin{algorithm}  [h]
\begin{algorithmic}[1]
 \caption{Algorithm for detecting variables with equivalent logical influence} 
 \State \textbf{Input}: $\varphi$
 \State \textbf{int} $k := 1$
 \State \textbf{Bool} $Equiv := \textit{false}$
 \ForEach {$a_i \in Var_{\varphi}$}
 \ForEach {$a_j \in (Var_{\varphi}\setminus a_i$)}
 \If  { $prog (\varphi, a_i = \top) = rename (prog (\varphi, a_j = \top), a_i, a_j) ~ \land $ 
 		 \State $prog (\varphi, a_i = \bot) = rename (prog (\varphi, a_j = \bot), a_i, a_j)$}
 	\State $E_k := \emptyset$
 	\State \textbf{add} $ a_j$ \textbf{to} $E_k$
 	\State \textbf{remove} $ a_j$ \textbf{from} $Var_{\varphi}$
	\State $Equiv := true$
 \EndIf
 \If{$Equiv = true$}
  	\State \textbf{add} $ a_i$ \textbf{to} $E_k$
	\State $Equiv := false$
  	\State $k++$
 \EndIf 
  \EndFor
 \EndFor
\label{alg:detectingEqVar}
\end{algorithmic}
\end{algorithm}

However, it is not possible to detect equivalent variables in large formulas
using the progression table due to the memory explosion problem
(i.e., the size of the table grows exponentially w.r.t. the number of variables).
It is therefore necessary to develop an algorithm that can be  used
to detect equivalent variables. Since we deal with formulas with Boolean variables
which take only two possible truth values, we can then develop an efficient
algorithm for detecting equivalent variables in a given formula as shown in Algorithm \ref{alg:detectingEqVar}. The algorithm takes advantage of the fact 
that the relation $\equiv$ is reflexive, symmetric, and transitive.

\section{Simplifications} \label{sec:simplificationRules}

When some variables are shown to be equivalent in their logical influences
w.r.t. the outcome of a formula, then some of these variables can be replaced by one representative. We now describe the basic steps that can be followed
to simplify a formula that contains variables
with equivalent logical influences.

\begin{enumerate}

\item Detect sets of variables in the formula
whose logical influences on the outcome of the formula are equivalent. 
This can be performed using Def. \ref{EquiVar}.

\item Fix the names of some variables (maybe 2-3 variables) in each derived set   
while replacing the names of the other variables to one of the fixed names.  

\item Reconstruct the formula using the new set of variables names.
This yields a formula with redundant variables.

\item Simplify the resulting formula by eliminating redundant variables.

\end{enumerate}

The resulting simplified LTL formula has the same syntactic structure
as the original formula  but in a reduced form, as the number of variables 
in the simplified formula is less than that of the original formula.

\begin{example}
Consider the following LTL formula
$$
\varphi =  G(a_1 \land a_2 \land ... \land a_{n_1}) \lor F (b_1 \land b_2 \land ... \land b_{n_2}).
$$
Clearly, the formula contains variables whose logical influences are equivalent.
To detect variables with equivalent  logical influences we use Definition \ref{EquiVar}. 
According to Definition \ref{EquiVar} the formula $\varphi$ has two sets of variables with equivalent  logical influences: $E_1 = \{a_1,..., a_{n_1}\}$ and $E_2 = \{b_1,..., b_{n_2}\}$.
Suppose that we choose to maintain the variables $a_1$ and $a_2$ from $E_1$
and replace the names of the other variables in $E_1$ by $a_1$ and $b_1$ and $b_2$
from $E_2$ and replace the names of the other variables in $E_2$ by $b_1$.
This yields the following formula
$$
\varphi^{'} =  G(a_1 \land a_2 \land a_1 \land... \land a_1) \lor F (b_1 \land b_2 \land b_1 \land... \land b_1).
$$

The formula $\varphi^{'}$ contains redundant variables and hence can be simplified to 
$$
\varphi_{R} = G(a_1 \land a_2) \lor F (b_1 \land b_2).
$$
In this case, we reduce the number of variables 
in the formula from $(n_1 + n_2)$ to 4 variables.
Such simplification helps to construct efficiently a simplification table
for $\varphi_R$ and draw some valid conclusions about $\varphi$ as
we shall discuss later.

\end{example}

The above described simplification rules lead to reduce the size
of formulas which contain variables with equivalent logical influences
from $n$ to $(n - (\sum_{i = 1}^{k} (|E_i|)) + 2 \times k)$,
where $n$ is the number of variables in the main formula,
and $k$ is the number of sets of variables whose logical influences are equivalent.

\section{From Simplified Formula to Original Formula} \label{sec:extensionRules}

We now describe the basic steps that can be followed 
to draw correct logical conclusions about the original formula
from the results obtained of the simplified formula.
Given an LTL formula $\varphi$ we simplify $\varphi$ 
to $\varphi_R$ by detecting and contracting variables with
equivalent logical influences as described at Section \ref{sec:simplificationRules}. 
Note that the simplified formula $\varphi_R$  contains only subset of the variables
of the original formula and hence conclusions
derived from the simplified formula need to be extended
while considering missing variables in the original formula
(i.e., variables that are in the original formula but not in the simplified formula).

\begin{enumerate}

\item Construct a progression table for the simplified formula $\varphi_R$.

\item Compute influence weights of the variables in the simplified formula $\varphi_R$.

\item Synthesize Boolean formulas for sets of configurations in the progression 
table of the formula $\varphi_R$ that yield the same LTL formula.

\item  Extend influence weights of the variables to the original formula $\varphi$.

\item Extend sets of synthesized Boolean formulas to the original formula $\varphi$.

\end{enumerate}

Note that steps (1-3) of the above procedure can be performed
as described in the previous section. 
We now describe how steps (4-5)
can be implemented by developing rules 
for extending logical conclusions  derived from the simplified formula.
Let $\mathbb{B}_{\phi}$ be a Boolean formula synthesized
from the progression table of the formula $\varphi_{R}$
for sets of configurations that yield the LTL formula $\phi$.
The general form of the  Boolean formula $\mathbb{B}_{\phi}^{\varphi_{R}}$ 
can be expressed as follows 
$$
\mathbb{B}_{\phi}^{\varphi_{R}} = (T_0 + T_1 + ... + T_n)
$$
where each term $T_i$ has the form $\prod (V)$ (a product of a set of variables),
where $V$ is a set of propositional variables from $prop(\varphi)$.
Let $\{E_1, ..., E_k \}$ be the sets of variables with equivalent logical
influence extracted from the formula $\varphi$.
Note that for each set $E_i$ we maintain only two variables
in the simplified formula.
Let us denote the variables maintained from the set $E_1$
 by $a_1$ and $a_2$ which we will use to formalize the extension rules given below.
Extending sets of Boolean formulas from the simplified formula
to the original formula can take one of the following forms:
(i)  extending $\mathbb{B}_{\phi}^{\varphi_{R}} $ by adding new 
variables to some terms in $\mathbb{B}_{\phi}^{\varphi_{R}} $, and
(ii)  extending $\mathbb{B}_{\phi}^{\varphi_{R}}$ by adding new terms to $\mathbb{B}_{\phi}^{\varphi_{R}} $.
The application of extension rules depends mainly
on the syntactic structure of the formula  $\mathbb{B}_{\phi}^{\varphi_{R}}$.

\begin{enumerate}

\item When none of the variables in the equivalent set $E_1$ appears in the formula $\phi$.
     That is, for all $a_i \in E_1$ we have $a_i \not \in prop(\phi)$.
     We have three cases here
      
      \begin{enumerate}
      
        
        \item if there exists a term $T_i$ in $\mathbb{B}_{\phi}^{\varphi_{R}}$
		such that $(|T.V| \geq 1 \land (T.V \cap E_1) = 1)$ then
		for each variable in $E_1$ that is not in the short formula	 $\varphi_{R}$	
		add a new term to $\mathbb{B}_{\phi}$ that is identical to $|T| $ 
		while replacing the variable $(T.V \cap E_1)$ by one from the
		set $E_1$ that is not in the short formula.
        
        \item if there exists a term $T_i$ in $\mathbb{B}_{\phi}^{\varphi_{R}}$
		such that $(|T.V| > 1 \land (T.V \cap E_1) = 2)$
		then add all variables in $E_1$ that is not in the short formula
		$\varphi_{R}$ to $V$.
        
        \item if none of the variables in $E_1$ appears in the terms of 
        	  $\mathbb{B}_{\phi}^{\varphi_{R}}$ then the formula  
        	  $\mathbb{B}_{\phi}^{\varphi_{R}}$ 
        	  needs not to be extended with respect to the set $E_1$.
        
      \end{enumerate}

\item When variables $a_1$ and $a_2$ appear in the formula $\phi$.
    We have two case here
     
  \begin{enumerate}
  

	\item   if  variables $a_1$ and $a_2$ appear in the formula $\phi$ 
	        but none of them appears in the terms of the
	         formula $\mathbb{B}_{\phi}^{\varphi_{R}}$. 
	        In this case, we need to extend
	        the formula $\phi$ by adding all variables in $E_1$ that are not in $\varphi_{R}$    		    to $\phi$ while preserving the syntactic structure of the formula $\phi$.

   \item  if variables $a_1$ and $a_2$ appear in the formula
	$\phi$ (i.e., $E_i \cap prop(\phi) = 2$) 
   and  in the formula $\mathbb{B}_{\phi}^{\varphi_{R}}$.
   Then the formula  $\mathbb{B}_{\phi}^{\varphi_{R}}$ will be extended in two steps
   (i) add all variables in $E_1$ that are not in $\varphi_{R}$ to $\phi$
    while preserving the syntactic structure of $\phi$,
   and (ii) use extension rules 1(a)-1(b) to extend the formula 
    $\mathbb{B}_{\phi}^{\varphi_{R}}$.

  \end{enumerate}
 
\end{enumerate}

\begin{theorem} \label{ExtensionTheorem}

Extension rules 1(a)-1(c) and 2(a)-2(b) are sound rules.

\end{theorem}

\begin{proof}

The proof of the Theorem can be constructed by case analysis,
where the shape (the syntactic structure) of the  Boolean formula
determines the way the formula will be extended.
Let $\varphi$ be an LTL formula and $\varphi_{R}$ 
be a simplified form of $\varphi$ obtained by detecting
and contracting equivalent variables in $\varphi$
as described at Section \ref{sec:simplificationRules}.
Let $E = \{a_1, a_2,..., a_n\}$ be a set of variables of $\varphi$ whose 
logical influences on the outcome of $\varphi$ are equivalent. 
Suppose that we maintain
two variables from $E$ in the simplified formula 
$\varphi_{R}$, let us denote them by $a_1$ and $a_2$.
Let $\mathbb{B}_{\phi}$
be a Boolean formula synthesized from the progression
table of $\varphi_{R}$ that we aim to extend to
the original formula $\varphi$.
Note first that the general form of the  Boolean formula $\mathbb{B}_{\phi}$ 
can be expressed as follows 
$$
\mathbb{B}_{\phi} = (T_0 + T_1 + ... + T_n)
$$
where each term $T_i$ has the form $\prod (V)$ (a product of a set of variables),
 $V$ is a set of propositional variables from $prop(\varphi)$, and $\phi$ is an LTL formula.
From the syntactic structure of the formula $\mathbb{B}_{\phi}$,
one can see that the extension of $\mathbb{B}_{\phi} $ to the original
formula can take one of the following forms:
(i)  extending $\mathbb{B}_{\phi} $ by adding new 
variables to some terms in $\mathbb{B}_{\phi} $, and
(ii)  extending $\mathbb{B}_{\phi} $ by adding new terms to $\mathbb{B}_{\phi} $.
The extension of  $\mathbb{B}_{\phi} $ depends on the way the variables $a_1$ and $a_2$
appear in $\mathbb{B}_{\phi} $, since the other variables that are
not appeared in the simplified formula are equivalent to variables $a_1$ and $a_2$ 
in their logical influences on the outcome of the formula.
There are two main cases to consider here

\begin{enumerate}

\item when none of the variables $a_1$ and $a_2$ appears in the formula $\phi$.
In this case, the extension of $\mathbb{B}_{\phi} $ depends on the appearance
of variables $a_1$ and $a_2$ in  $\mathbb{B}_{\phi} $. 
The extension will be proceeded in an iterative way
by examining the terms of the formula $\mathbb{B}_{\phi} $.
For this case, there are several sub-cases to consider 

\begin{enumerate}

\item if there exists a term $T_i$ in $\mathbb{B}_{\phi} $
where both variables $a_1$ and $a_2$ are in $T_i.V$
then the variable $a_3$ must be added to the list $V$,
where $a_3$ is a variable in the original formula but not
in the simplified formula whose logical influence
to outcome of the formula is equivalent to $a_1$ and $a_2$.
It is easy to see the soundness of this rule
as $a_3 \equiv a_1 \equiv a_2$.

\item if there exists a term $T_i$ in $\mathbb{B}_{\phi} $
where only variable $a_1$ or $a_2 $ appears in $T_i.V$.
Then a new term will be added to  $\mathbb{B}_{\phi} $
 with the same syntactic structure as  $T_i$
 while replacing the instance of $a_1$ or $a_2 $
 by $a_3$. Again this is due to the observation 
 that $a_3 \equiv a_1 \equiv a_2$ and hence they 
 have the same logical influence on the outcome
 of the formula.

\item if neither $a_1$ nor $a_2$ appears in any of the terms in $\mathbb{B}_{\phi} $.
Then obviously none of the missing variables that have equivalent 
influence on the outcome of the  formula $\varphi$ 
will appear in the terms of $\mathbb{B}_{\phi} $.

\end{enumerate}


\item when variables $a_1$ and $a_2$ appear in the formula $\phi$.
There are two cases here

\begin{enumerate}

\item if $a_1$ and $a_2$ appear in $\phi$ but none of them appear in $\mathbb{B}_{\phi} $. 
In this case the variable $a_3$ must be added to formula $\phi$,
where $a_3$ is a variable in the original formula but not
in the simplified formula whose logical influence
on outcome of the formula is equivalent to $a_1$ and $a_2$.
However, since  $a_1$ and $a_2$ do not appear in $\mathbb{B}_{\phi} $
then $\mathbb{B}_{\phi} $ needs not to be extended w.r.t. $E$.

\item if $a_1$ and $a_2$ appear in $\phi$ and appear in $\mathbb{B}_{\phi} $. 
In this case the formula $\mathbb{B}_{\phi} $ will be extended
into two steps: (i) variable $a_3$ must be added to formula $\phi$,
and (ii) the terms of $\mathbb{B}_{\phi} $ will be extended
using rules 1(a) and 1(b). This is mainly
because $a_3 \equiv a_1 \equiv a_2$ and hence they have equivalent logical influence
on the outcome of the main formula.

\end{enumerate}

\end{enumerate}

\end{proof}

Note that we synthesize Boolean formulas only
for sets of configurations in formula progression table
that yield same LTL formulas, and hence lot of  
configurations will not be considered when extending formulas.
We consider here the cases that maybe 
encountered during analysis. 
The extension rules are in general
straightforward rules as we deal with 
variables whose logical influences on the outcome
are equivalent.
However, one may need to develop further rules depending
on the syntactic structure of synthesized
Boolean formulas from simplified formula.
We now discuss some basic properties of influence weights of variables
and some useful lemmas that can be
used to simplify the computation of influence weights
of variables in formulas with large number of variables.

\begin {definition} (\textbf{Properties of influence weights of variables}.)
Let $\varphi$ be an LTL formula with set of variables 
$prop(\varphi)  = \{a_1,..., a_n\}$. 
The basic properties of logical influence weights 
of $ \{a_1,..., a_n\}$ can be summarized as follows

\begin{enumerate}

\item for any variable $a_i \in prop(\varphi) $ we have $0 \leq IW_{\varphi} (a_i) \leq 1$. 

\item  when $a_i \equiv a_j$ then $IW_{\varphi} (a_i) = IW_{\varphi} (a_j)$ but the converse in not true. 

\item when $IW_{\varphi} (a_i) > IW_{\varphi} (a_j)$ we say that the variable $a_i$ 
has higher logical influence on the outcome of $\varphi$ than the variable $a_j$.

\item when $IW_{\varphi} (a_i) = 1$ we say that $a_i$ is a variable
of weight one in the sense that a definite truth value of $\varphi$ 
cannot be obtained without knowing $a_i$. 

\end{enumerate}

\end{definition}

Variables of weight one are key variables
in the formula as satisfaction/falsification
cannot be determined without knowing their truth values.
Therefore, variables of weight one should
receive higher priority than variables of weight less than one
 when considering solutions that are sensitive to variable ordering.
 For example the size of a Boolean Decision Diagram (BDD) for a given Boolean 
 function is sensitive to the ordering of the variables in the BDD.

\begin{lemma} \label{VarOfWeightOne}

Let $\varphi$ be an LTL formula with a set of propositional variables $prop(\varphi) = \{a_1,..., a_n\}$.
Let also $\varphi_{R}$ be a simplified version of $\varphi$ computed
as described at Section \ref{sec:simplificationRules}.
Then when $IW_{\varphi_{R}} (a_i) = 1$  we have $IW_{\varphi} (a_i) = 1$ as well.
\end{lemma}

\begin{proof}

From the definition of influence weighs of variables (Definition \ref{InfluenceWight})
and that $IW_{\varphi_{R}} (a_i) = 1$ we notice that the presence of all variables
in the formula $\varphi_{R}$ do not affect the weight of the variable $a_i$.
That is, the variable $a_i$ appears in all simplified formulas in the progression
table of $\varphi_{R}$ that result from the truth combinations in which $a_i = ?$. Note that the formula $\varphi_{R}$ is a simplified version of
$\varphi$ ($\varphi_{R}$ has the same syntactic structure of $\varphi$ but in a short from) in which two variables from each list of variables with equivalent logical influence from the  formula $\varphi$ are maintained. 
Let $E_k$ be a list of variables with equivalent logical influence
derived from the formula $\varphi$ and that variables $b_1, b_2 \in E_k$ 
have been chosen to be  maintained in $\varphi_{R}$. 
It is easy to see that adding any new variable $b_j \in E_k$ to $\varphi_{R}$
such that $b_j \equiv b_1 \equiv b_2$ will not affect the influence weight of $a_i$
as $b_3$  has equivalent logical influence to $b_1$ and $b_2$ 
and hence  $IW_{\varphi} (a_i) = 1$.

\end{proof}

\begin{theorem} \label{VarOfWeightLessOne}

Let $\varphi$ be an LTL formula with a set of propositional variables $prop(\varphi)  = \{a_1,..., a_n\}$.
Let $\varphi_{R}$ be a simplified formula of $\varphi$ computed
as described at Section \ref{sec:simplificationRules}.
Suppose that all variables in $\varphi$ have  equivalent logical influence 
on the outcome of $\varphi$ and that $IW_{\varphi_{R}} (a_1) =  \frac{N}{D}$,
where $N$ is the denominator of the fraction and $D$ is the denominator of the fraction. 
Then $ IW_{\varphi} (a_1)= \frac{N^{n-1}}{D^{n-1}}$.

\end{theorem}

\begin{proof}

From the definition of influence weighs of variables 
(Definition \ref{InfluenceWight}) we know that $IW_{\varphi_{R}} (a_i)$
is a fraction of the form  $\frac{N}{D}$ and that $N \leq D$,
 where the numerator  $N$ represents
the number of formulas in the progression table of the formula $\varphi_{R}$
that $a_1$ appears in and the denominator $D$ 
represents the number of truth combinations of the variables
 of $\varphi_{R}$ in which $a_1 = ?$.
Note that the denominator of the fraction has always the form $3^{n-1}$ 
and hence $D = 3$ regardless of the syntactic structure of the formula.
This is mainly because the variables 
take their truth values from the truth domain $\mathbb{B}_3 = \{\bot, \top, ?\}$.
Note that since all variables in $\varphi$ have equivalent logical
influence on the outcome of the formula 
then $\varphi$ can be simplified to a formula $\varphi_{R}$
with only two variables, let us denote them by $a_1$ and $a_2$. 
However, since the progression table grows exponentially
w.r.t. the number of variables and that all variables in the formula $\varphi$ 
have equivalent logical influence (i.e., $a_1 \equiv a_2 \equiv ... \equiv a_n$)
then it is easy to see that $ IW_{\varphi} (a_1)= \frac{N^{n-1}}{3^{n-1}}$,
where $N$ is the numerator of $IW_{\varphi_{R}} (a_1)$ 
and $n$ is the number of variables in $\varphi$.

\end{proof}






Lemma \ref{VarOfWeightOne} states that variables of weight one
do not get influenced by adding more variables to the formula
as long as the syntactic structure of the formula is preserved.
On the other hand, Theorem \ref{VarOfWeightLessOne} states 
that for formulas whose variables are equivalent in their
logical influences then the influence weights of these variables 
can be computed in  a straightforward way using the formula 
 $ IW_{\varphi} (a_i)= \frac{N^{n-1}}{3^{n-1}}$,
 where $a_i$ is a variable in $\varphi$,
 $N$ is the numerator of the fraction $IW_{\varphi_{R}} (a_1)$
 and $n$ is the number of variables in the original formula $\varphi$.

 \begin{example}

Consider the following LTL formula
$$
\varphi = F (a_1 \land a_2 \land a_3 \land a_4 \land a_5) \lor G (b_1 \land b_2 \land b_3 \land b_4)
$$
Note that $\varphi$ has two sets of variables with equivalent logical behavior:
$E_1 = \{a_1, a_2, a_3, a_4, a_5 \}$ and $E_2 = \{b_1, b_2, b_3, b_4 \}$.
Using the simplification rules described at Section \ref{sec:simplificationRules}
we can simplify $\varphi$ to $\varphi_R = F (a_1 \land a_2) \lor G (b_1 \lor b_2)$.
The progression table of the reduced formula is given in Table \ref{table:EX2}.
We consider here the Boolean formulas for the the cases
of configurations that lead to the simplified formulas $XF (a_1 \land a_2)$ and $\top$.
The expressions can be given as follows
$$
\mathbb{B}_{(XF (a_1 \land a_2))}^{\varphi_{R}}   = \sum_{i =1..2, j = 1..2} (\overline{a_i}.\overline{b_j}) \hspace*{40pt} 
\mathbb{B}_{\top}^{\varphi_{R}} = \prod_{i=1..2} (a_i)
$$

Extending the Boolean expression $\mathbb{B}_{(XF (a_1 \land a_2))}^{\varphi_{R}} $ to the original formula can be performed using rule 2(c), while extending the expression $\mathbb{B}_{\top}^{\varphi_{R}} $ to the original formula
can be performed using rule 1(b) which yield the following formulas 

$$
\mathbb{B}_{(XF (a_1 \land a_2 \land a_3 \land a_4 \land a_5))}^{\varphi}   = \sum_{i =1..5, j = 1..4} (\overline{a_1}.\overline{b_j}) \hspace*{30pt} 
\mathbb{B}_{\top}^{\varphi}  = \prod_{i=1..5} (a_i)
$$

Note that the influence weights of the variables $a_1, a_2, a_3, a_4,$ and $a_5$
will be the same since their logical influences on the outcome of the formula are equivalent.
From the progression table of the simplified formula
we note that $IW_{\varphi_{R}} (a_1) = IW_{\varphi_{R}} (a_2) = 1$
and $IW_{\varphi_{R}} (b_1)= IW_{\varphi_{R}} (b_2) = 0.66$. 
From Lemma \ref{VarOfWeightOne} we conclude that 
$IW_{\varphi} (a_1) = IW_{\varphi} (a_2) = 1$
and from Theorem \ref{VarOfWeightOne} we conclude that
$IW_{\varphi} (b_1)= IW_{\varphi} (b_2) \approx 0.039$.  
\end{example}

\section{Using Progression Table in Decentralized  Monitoring} \label{sec:DRVAlgori}

The great challenge in developing efficient decentralized 
framework for distributed systems is to decide: 
(i) which process communicates to which,
 (ii) when they communicate, and (iii) what they communicate.
 To address these challenges when monitoring a formula,
 we construct first a  progression
table for the monitored formula from which
we compute the influence weights of each 
variable in the formula
and derive some Boolean formulas
for the sets of configurations
that yield the same simplified LTL formula.
The extracted information is used for two purposes: 
(i) to synthesize efficient communication strategy for processes,
and (ii) to propagate observations of processes in an efficient way.
For each process, 
we associate what we call process influence logical factor.
Such factor can be computed according to the observation power of the process
(i.e., the set of variables in the formula that are locally observable by the process).

\begin{definition} \label{InfluenceFactor} (\textbf{Influence factors of processes}.)
Let $P$ be a distributed system with $n$ processes $\{p_0,.., p_{n-1} \}$
and $\varphi$ be an LTL property of $P$ that we seek to monitor in a decentralized fashion. 
Let $p_i \in P$ be a process with a set of atomic propositions $AP_i = \{a_1,..., a_k\}$
and that $AP_i \subseteq prop (\varphi)$.
The influence factor of process $p_i$ (denoted as $IF_{\varphi} (p_i) $) 
can be computed as follows
$$
IF_{\varphi} (p_i) = \sum_{j = 1}^{k} (IW_{\varphi} (a_j)).
$$
That is, the influence logical factor of a process can be computed
by taking the sum of the logical weights of the variables observable by that process.
\end{definition}

Using Definition \ref{InfluenceFactor} we can then
synthesize an efficient round-robin communication policy 
for processes according to their observation power.
In our setting, processes with higher influence factor will receive higher
priority in the order of communication.
This is mainly because processes with higher influence factors
they either observe larger number of variables
of the  monitored formula or variables with higher
influence weights and hence their ability to simplify
the formula are higher than those with lower influence factors.

\begin{example}

Suppose that we would like to monitor a formula 
$\varphi = F (b \lor (a_1 \land a_2 \land c))$
and that we have three processes: process $A$ with $AP_A = \{a_1, a_2\}$,
 process $B$ with $AP_B = \{b\}$, and process $C$ with $AP_C = \{c\}$. 
To synthesize an efficient round-robin communication policy for processes
we use Definition \ref{InfluenceFactor} to compute their influence factors.
We first need to compute the logical influence weight of each variable in the formula.
This can be computed by constructing a progression table for the formula $\varphi$.
From the progression table of the formula we find that
$
IW_{\varphi} (a_1) = IW_{\varphi} (a_2) =IW_{\varphi} (c) = \frac{8}{27}$
 and $IW_{\varphi} (b) = \frac{26}{27}.
$
From these values we can see that the influence factors of processes are: $IF_{\varphi} (A) = \frac{16}{27}$, $IF_{\varphi} (B) = \frac{26}{27}$, and $IF_{\varphi} (C) = \frac{8}{27}$. 
However, since $IF_{\varphi} (B) > IF_{\varphi} (A) > IF_{\varphi} (C)$
then the round-robin policy will be of the form
$
(B \rightarrow A \rightarrow C \rightarrow B)
$,
where the direction of the arrows represents the order of communication.

\end{example}

We now turn to discuss how processes propagate their observations
during runtime verification.
Instead of allowing processes to propagate their entire
observations to their neighbor processes, 
they can take advantage of the constructed progression table of the
formula to compute the minimal set of variables whose truth values
need to be propagated. Note that in some
situations it is sufficient for processes to propagate only a subset of their
observations while allowing the receiving process 
to draw the same conclusion about the truth value of the monitored formula. 
Suppose for example that processes $A$ and $B$ monitor an LTL
formula $\varphi = F (a_1 \land a_2 \land b_1 \land b_2)$
and that process $A$ observes $a_1$ and $a_2$. 
Suppose that at some state $s$ process $A$ observes that $a_1 = \bot \land a_2 = \top$.
Then $A$ needs only to propagate the truth value of $a_1$
to $B$ as this would be sufficient to allow $B$ to know 
that $(a_1 \land a_2 \land b_1 \land b_2) = \bot$
and hence $ F (a_1 \land a_2 \land b_1 \land b_2) = ?$.

The advantage of synthesizing Boolean formulas
characterizing the conditions under which the
monitored formula can be simplified to certain
formulas is that they can be used to compute the minimal set
of variables whose truth values 
need to be propagated.
As mentioned earlier, a Boolean formula is given as sums of products
of the form $B_{\phi} = (T_0 + T_1+ ... + T_k)$,
where each term $T_i$ represents a condition under which
the formula $\varphi$ can be simplified to $\phi$ and
has the form $\prod (V)$ where $V$ is a set of variables.
Suppose that at some step $s$ of the trace being monitored
process $A$ simplifies the monitored formula $\varphi$
to formula $\phi$ using its observations.
The question is then what $A$ should
communicate to its neighbor process (i.e.,
which variables whose truth values need to be propagated)?
A simple procedure for computing
the minimal set of variable can be developed by examining sets of synthesized 
Boolean formulas as described below. 

\begin{enumerate}

\item Find all terms in the formula $B_{\phi}$ which
hold to \textit{true} when replacing 
the variables in $B_{\phi}$ by their definite truth values. 
Let us denote the set containing all the terms that hold to \textit{true}
in the formula  $B_{\phi}$  by $L$.

\item Find the term in $L$ with the smallest corresponding $V$ set, 
let us denote that set by $V_{min}$.
In this case, the variables in the set $V_{min}$
represent the minimal set of variables whose
truth values need to be propagated.

\end{enumerate}

Our decentralized monitoring algorithm consists of two phases: setup and monitor.
The setup phase consists of the five steps described at Section \ref{sec:simplificationRules}.
We now summarize the actual monitoring steps in the form of an explicit algorithm that describes how local monitors operate and make decisions:

\begin{enumerate}

\item $[$Read next event$]$. Read next $\sigma_i \in \Sigma_i$ (initially each process reads $\sigma_0$).

\item $[$Compute minimal set of variables to be transmitted$]$. Examine the set of Boolean
formulas derived from the progression table to compute the minimal set of variables 
whose truth values need to be propagated.

\item $[$Compute the receiving process$]$. For our communication strategy, the receiving
process of some process $p$ is fixed between states and computed according
to some round-robin communication policy, as described in Section \ref{sec:DRVAlgori}.

\item $[$Propagate truth values of variables in $V_{min}$ $]$. Propagate the truth values of variables in the minimal set in  $V_{min}$ to the receiving process.


\item $[$Evaluate the formula $\varphi$ and return$]$.  If a definite verdict of $\varphi$ is found return it. That is, if $\varphi = \top$ return $\top$, if  $\varphi = \bot$ return $\bot$.

\item 	$[$Go to step 1$]$. If the trace has not been finished or a decision has not been made then go to step 1.

\end{enumerate}

We now turn to discuss the basic properties of our decentralized monitoring framework. 
Let $\models_{D}$ be the satisfaction relation on finite traces in the decentralized setting and $\models_{C}$ be the satisfaction relation on finite traces in the centralized setting, where both $\models_{D}$ and $\models_{C}$ yield values from the same truth domain. Note that in a centralized monitoring algorithm we assume that there is a central process that observes the entire global trace of the system being monitored, while in our  decentralized monitoring algorithm processes observe part of the trace, perform remote observation, and use the progression
table of the monitored formula in order to setup an efficient communication
strategy and to propagate observations in an optimal way. The following theorems stating the soundness and completeness of our decentralized monitoring algorithm.

\begin{theorem} (\textbf{Soundness}). \label{soundness}
Let $\varphi \in LTL$ and $\alpha \in \Sigma^{*}$. Then $\alpha \models_{D} \varphi =\top/\bot \rimp \alpha \models_{C} \varphi =\top/\bot$.

\end{theorem}

Soundness means that all verdicts (truth values taken from a truth-domain) found by the decentralized monitoring algorithm for a global trace $\alpha$ with respect to the property $\varphi$ are actual verdicts that would be found by a centralized monitoring algorithm that have access to the trace $\alpha$. 

\begin{theorem} (\textbf{Completeness}).
Let $\varphi \in LTL$ and $\alpha \in \Sigma^{*}$. Then $\alpha \models_{C} \varphi =\top/\bot \rimp \alpha \models_{D} \varphi =\top/\bot$.

\end{theorem}

Completeness means that all verdicts found by the centralized monitoring algorithm 
for some trace $\alpha$ with respect to the property $\varphi$ 
will eventually be  found by the decentralized monitoring algorithm. 
The soundness and completeness of our monitoring approach can be inferred from the 
soundness of the progression table of a formula, Theorem \ref{ExtensionTheorem},
and the round-robin strategy.

\section{Experiments}

We have evaluated our monitoring approach
against the LTL decentralized monitoring 
approach of Bauer and Falcone \cite{BauerF12},
in which the authors developed a monitoring algorithm for LTL based 
on the formula-progression technique \cite{Bacchus1996}.
The formula progression technique takes a temporal formula $\phi$ and a current assignment $I$ over the literals of $\phi$ as inputs and returns a new formula after acting  $I$ on $\phi$. The idea is to rewrite a temporal formula when an event $e$ is observed or received to a formula which represents the new requirement that the monitored system should fulfill for the remaining part of the trace.
We also use the tool DECENTMON3 \texttt{(http://decentmon3.forge.imag.fr/)} in our evaluation,
which is a tool dedicated to decentralized monitoring.
The tool takes as input multiple traces,
 corresponding to the behavior of a distributed system, and an LTL formula.
 The reason for choosing DECENTMON3 in our evaluation is that 
 it makes similar assumptions to our presented approach.
 Furthermore, DecentMon3 improves the original DecentMon tool developed in \cite{BauerF12}
  by limiting the growth of the size of local obligations 
  and hence it may reduce the   size of propagated messages. 
  We believe that by choosing the tool DECENTMON3 as baseline for comparison 
  we make the evaluation much fairer.

We denote by BF the monitoring approach of Bauer and Falcone,
and PDM our presented approach in which processes 
construct a progression table for the monitored formula
which will be used to synthesize efficient round
robin policy for processes and to propagate
observations in an optimal way.
We compare the approaches against 
benchmark for patterns of formulas \cite{patternSite} (see Table \ref{table: tableResult2}). 
In Table \ref{table: tableResult2}, the following metrics are used: $\# msg$, the total number of exchanged messages; $|msg|$, the total size of exchanged messages (in bits);  $|trace|$,  the average length of the traces needed to reach a verdict; and $|mem|$, the memory in bits needed for the structures (i.e., formulas plus state for our algorithm).
 For example, the first line in Table \ref{table: tableResult2} 
 says on average, traces were of length 4.65 when  one 
of the local monitors in approach BF came to a verdict, 
and of length 5.26 when one of the monitors in PDM  came to a verdict.

\begin{table*} 
\begin{center}
\begin{tabular}{ |c|c|c|c|c|c|c|c|c| }
\hline
 &  \multicolumn{2}{|c|}{$|trace|$}& \multicolumn{2}{|c|}{$\# msg.$} & \multicolumn{2}{|c|}{$|msg.|$}  & \multicolumn{2}{|c|}{$|mem|$} \\
 \hline
$|\varphi|$ & BF  & PDM & BF & PDM & BF & PDM& BF & PDM  \\
\hline
abs &  4.65   & 5.10 & 4.46  & 5.15&  1,150    & 102  & 496 .4 & 11.9 \\
\hline
exis &  27.9 &  29.5  &  19.7 & 20.8 &  1,100 & 411  & 376  & 19.8\\
\hline
bexis &  43.6 & 41.3 &  31.6 &   31.9 &  55,000 & 25415  & 28,200&  20.6\\
\hline
univ &   5.86&  6.2 &   5.92&  5.82 &   2,758 & 138  & 498& 22.5\\
\hline
prec &   54.8 & 54.5 &   25.4& 26.9 &   8,625 & 755  & 663 & 34.9\\
\hline
resp &   622 & 622 &   425&  515 &    22,000 & 1211 & 1,540 & 17.5\\
\hline
precc & 4.11 & 5.2 &  4.81& 5.95 & 5,184 & 356 & 1,200 & 15.7\\
\hline
respc & 427 & 444&  381 & 409& 9,000 & 2799&4,650 & 22.1 \\
\hline
consc & 325 & 324& 201 & 234 & 7,200 & 1223 & 2,720 & 15.8 \\
\hline
\end{tabular}
\end{center}
\caption{Benchmarks for 1000 generated LTL pattern formulas (Averages)} \label{table: tableResult2}
\end{table*}

\subsection{Benchmarks for Patterns of formulas}

We  compared the two approaches with realistic
specifications obtained from specification patterns \cite{Dwyer1999}.
Table \ref{table: tableResult2} reports the verification results
for different kinds of patterns 
(absence, existence, bounded existence, universal, precedence,
response, precedence chain, response chain, constrained chain).
The specification formulas are available at \cite{patternSite}.
We generated 1000 formulas monitored over
the same setting (processes are synchronous and reliable).
For this benchmark we generated formulas
as follows. For each pattern, we
randomly select one of its associated formulas.
Such a formula is ``parametrized''
by some atomic propositions from the alphabet of the distributed system
which are randomly instantiated.
For this benchmark (see Table \ref{table: tableResult2}), 
the presented approach  leads to significant reduction
on both the size of messages and the amount of memory consumption
compared to the optimized version of BF algorithm (DECENTMON3).

\section{Related Work}

Finding redundancies in formulas has been studied in the form of
vacuity detection in temporal logic formulas \cite{Kupferman1999,ArmoniFFGPTV03}. 
Here, the goal is to identify  vacuously valid subparts of formulas, 
indicating, for example, a specification error.
In contrast, our focus is to reduce the complexity of the formula by detecting
variables whose logical influences on the outcome of the formula are equivalent
and then reduce the complexity of the formula by reducing number of variables.
The goal is to analyze efficiently a simplified form of the formula 
and draw some correct conclusions about the original formula
by applying some valid extension rules.

The problem of representing formulas compactly has received attention from
many different angles. For example, BDDs attempt to represent propositional
formulas concisely, but they suffer from the variable ordering problem and are
prone to a worst-case exponential blow-up \cite{Bryant1992}. 
We believe that our approach based on formula simplification table
 can be used to optimize dramatically BDD construction of Boolean formulas,
as it helps to identify variables with highest
logical influences on the outcome of the formula.
Furthermore, computing logical influence weights of variables
in a formula can help to find an optimal variable ordering
which can lead to the most compact representation of the formula.


Various simplification rules have also been successfully applied
as  a  preprocessing  step  for  solving,  usually  for  bit-vector  arithmetic  \cite{Ganesh2007,Jha2009}.
These  rules  are  syntactic  and  theory-specific. 
In contrast, the technique described in this paper is not meant as a
preprocessing step for solving and guarantees non-redundancy,
it is rather a simplification technique for detecting and contracting
variables with equivalent logical influences
for the purpose of optimizing formal analysis of formulas
by constructing simpler forms sufficient to prove the original property.

The literature on decentralized monitoring problem
is a rich literature, where several monitoring algorithms
 have been developed for verifying distributed systems at runtime \cite{Sen2004,BauerF12,ColomboF14,FalconeCF14,Scheffel14,MostafaB15}.
We discuss here some interesting works on the problem and
 refer the reader to \cite{Hokayem2017,Bataineh18} for a more comprehensive survey.

Bauer and Falcone \cite{BauerF12} propose a decentralized framework for runtime monitoring of LTL. The framework is constructed from local monitors which can only observe the truth value of a predefined subset of propositional variables. The local monitors can communicate their observations in the form of a (rewritten) LTL formula towards its neighbors. 
Mostafa and Bonakdarpour \cite{MostafaB15} propose similar decentralized LTL monitoring framework, but truth value of propositional variables rather than rewritten formulas are shared.

The work of Falcone et al. \cite{FalconeCF14} proposes a general decentralized monitoring algorithm in which the input specification is given as a deterministic finite-state automaton rather than an LTL formula. Their algorithm takes advantage of the semantics of finite-word automata, and hence they avoid the monitorability issues induced by the infinite-words semantics of LTL. They show that their implementation outperforms the Bauer and Falcone decentralized LTL algorithm \cite{BauerF12} using several monitoring metrics.

 Colombo and Falcone \cite{ColomboF16} propose a new way of organizing monitors called choreography, where monitors are organized as a tree across the distributed system, and each child feeds intermediate results to its parent. The proposed approach tries to minimize the communication induced by the distributed nature of the system and focuses on how to automatically split an LTL formula according to the architecture of the system.  

El-Hokayem and Falcone \cite{Hokayem2017} propose a new framework for decentralized monitoring with new data structure for symbolic representation and manipulation of monitoring information in decentralized monitoring. In their framework, the formula is modeled as an automaton where transitions of the automaton are labeled with Boolean expressions over atomic propositions of the system.

Recently, Al-Bataineh and Rosenblum \cite{Bataineh18} propose a new framework 
for decentralized LTL monitoring based on the notion of the tableau technique,
where the monitored formula is represented and decomposed using tableau
decomposition rules. In their framework,
they develop also a logical inference engine that allows processes
to propagate their observations as truth values of
 atomic formula, compound formulas, and temporal
 formula which depends mainly on the syntactic structure of the formula
 and the observation power of processes.

\section{Conclusion and Future Work}

We presented a novel framework for decentralized monitoring
of LTL formulas based on the notion
of formula progression table.
The formula progression table is a mathematical table that
shows all possible resulting forms of the main formula under 
different truth assignments of its variables.
The progression table can be used to extract several useful
information about the analyzed formula
including logical influence weights of the variables in the formula
and to identify variables with highest logical influence
on the outcome of the formula.
We showed how formula progression table can be used
to optimize decentralized monitoring solutions of LTL formulas
by synthesizing efficient communication strategies for processes
and propagating information in an optimal way.
In future work, we aim to employ some decomposition techniques
to split the global LTL formula into local LTL expressions. 
This would allow processes to
construct multiple progression tables,
which help to avoid the memory-explosion problem 
 of the formula progression table when dealing with large LTL formulas.

 \bibliographystyle{plain}
\bibliography{references}

\begin{thebibliography}{10}

\bibitem{Bataineh18}
Omar~I. Al{-}Bataineh and David Rosenblum.
\newblock Efficient decentralized {LTL} monitoring framework using tableau
  approach.
\newblock {\em CoRR}, abs/1803.02051, 2018.

\bibitem{patternSite}
H.~Alavi, Avrunin, J.~G., Corbett, L.~Dillon, M.~Dwyer, and C.: Pasareanu.
\newblock Specification patterns website.
\newblock \url{http://patterns.projects.cis.ksu.edu/}, 2011.

\bibitem{ArmoniFFGPTV03}
Roy Armoni, Limor Fix, Alon Flaisher, Orna Grumberg, Nir Piterman, Andreas
  Tiemeyer, and Moshe~Y. Vardi.
\newblock Enhanced vacuity detection in linear temporal logic.
\newblock In {\em Computer Aided Verification, 15th International Conference,
  {CAV}}, pages 368--380, 2003.

\bibitem{Bacchus1996}
Fahiem Bacchus and Froduald Kabanza.
\newblock Planning for temporally extended goals.
\newblock In {\em Proceedings of the Thirteenth National Conference on
  Artificial Intelligence}, pages 1215--1222, 1996.

\bibitem{Bauer2011}
Andreas Bauer, Martin Leucker, and Christian Schallhart.
\newblock Runtime verification for ltl and tltl.
\newblock {\em ACM Transactions on Software Engineering and Methodology
  (TOSEM)}, pages 14:1--14:64, 2011.

\bibitem{BauerF12}
Andreas~Klaus Bauer and Yli{\`{e}}s Falcone.
\newblock Decentralised {LTL} monitoring.
\newblock In {\em {FM} 2012: Formal Methods - 18th International Symposium,
  Paris, France}, pages 85--100, 2012.

\bibitem{Bryant1992}
Randal~E. Bryant.
\newblock Symbolic boolean manipulation with ordered binary-decision diagrams.
\newblock {\em ACM Computing Surveys}, pages 293--318, 1992.

\bibitem{ColomboF14}
Christian Colombo and Yli{\`{e}}s Falcone.
\newblock Organising {LTL} monitors over distributed systems with a global
  clock.
\newblock In {\em Runtime Verification - 5th International Conference, {RV}
  2014}, pages 140--155, 2014.

\bibitem{ColomboF16}
Christian Colombo and Yli{\`{e}}s Falcone.
\newblock Organising {LTL} monitors over distributed systems with a global
  clock.
\newblock {\em Formal Methods in System Design}, 49(1-2):109--158, 2016.

\bibitem{Dwyer1999}
Matthew~B. Dwyer, George~S. Avrunin, and James~C. Corbett.
\newblock Patterns in property specifications for finite-state verification.
\newblock In {\em Proceedings of the 21st International Conference on Software
  Engineering}, pages 411--420, 1999.

\bibitem{Hokayem2017}
Antoine El-Hokayem and Yli\`{e}s Falcone.
\newblock Monitoring decentralized specifications.
\newblock In {\em Proceedings of the 26th ACM SIGSOFT International Symposium
  on Software Testing and Analysis (ISTA)}, pages 125--135, 2017.

\bibitem{FalconeCF14}
Yli{\`{e}}s Falcone, Tom Cornebize, and Jean{-}Claude Fernandez.
\newblock Efficient and generalized decentralized monitoring of regular
  languages.
\newblock In {\em Formal Techniques for Distributed Objects, Components, and
  Systems}, pages 66--83, 2014.

\bibitem{Ganesh2007}
Vijay Ganesh and David~L. Dill.
\newblock A decision procedure for bit-vectors and arrays.
\newblock In {\em Proceedings of the 19th International Conference on Computer
  Aided Verification}, CAV'07, pages 519--531, 2007.

\bibitem{Jha2009}
Susmit Jha, Rhishikesh Limaye, and Sanjit~A. Seshia.
\newblock Beaver: Engineering an efficient smt solver for bit-vector
  arithmetic.
\newblock In {\em Proceedings of the 21st International Conference on Computer
  Aided Verification}, CAV '09, pages 668--674, 2009.

\bibitem{Kupferman1999}
Orna Kupferman and Moshe~Y. Vardi.
\newblock Vacuity detection in temporal model checking.
\newblock In {\em Proceedings of the 10th IFIP WG 10.5 Advanced Research
  Working Conference on Correct Hardware Design and Verification Methods},
  pages 82--96, 1999.

\bibitem{MostafaB15}
Menna Mostafa and Borzoo Bonakdarpour.
\newblock Decentralized runtime verification of {LTL} specifications in
  distributed systems.
\newblock In {\em 2015 {IEEE} International Parallel and Distributed Processing
  Symposium}, pages 494--503, 2015.

\bibitem{Pnueli1977}
Amir Pnueli.
\newblock The temporal logic of programs.
\newblock In {\em Proceedings of the 18th Annual Symposium on Foundations of
  Computer Science}, SFCS '77, pages 46--57. IEEE Computer Society, 1977.

\bibitem{Scheffel14}
Torben Scheffel and Malte Schmitz.
\newblock Three-valued asynchronous distributed runtime verification.
\newblock In {\em International Conference on Formal Methods and Models for
  System Design (MEMOCODE)}, volume~12. IEEE, 2014.

\bibitem{Sen2004}
Koushik Sen, Abhay Vardhan, Gul Agha, and Grigore Rosu.
\newblock Efficient decentralized monitoring of safety in distributed systems.
\newblock In {\em Proceedings of the 26th International Conference on Software
  Engineering}, ICSE '04, pages 418--427. IEEE Computer Society, 2004.

\end{thebibliography}

\end{document}